\documentclass{article}
\usepackage[margin=1in]{geometry}
\renewcommand\appendix{\par
  \setcounter{section}{0}
  \setcounter{subsection}{0}
  \setcounter{figure}{0}
  \setcounter{table}{0}
  \renewcommand\thesection{Appendix \Alph{section}}
  \renewcommand\thefigure{\Alph{section}\arabic{figure}}
  \renewcommand\thetable{\Alph{section}\arabic{table}}
}
\usepackage[utf8]{inputenc}
\usepackage{verbatim}
\usepackage{fullpage}
\usepackage{amsthm}
\usepackage{amsmath}
\usepackage{graphicx}
\newtheorem{thm}{Theorem}[section]
\newtheorem{lem}[thm]{Lemma}
\usepackage{authblk}
\usepackage{algorithm,algorithmic}
\usepackage[T1]{fontenc}
\theoremstyle{definition}
\newtheorem{defn}{Definition}[section]
\newtheorem{exmp}{Example}[section]
\newcommand\blfootnote[1]{%
  \begingroup
  \renewcommand\thefootnote{}\footnote{#1}%
  \addtocounter{footnote}{-1}%
  \endgroup
}

\title{Group Symmetries of Complementary Code Matrices}
\author[1]{Brooke Logan\thanks{brookelogan974@gmail.com}\thanks{Rowan University's Mathematics Department funded the research during Summer 2013.  Bantivolglio Honors Concentration funded the research during Spring 2014. }}
\author[1]{Hieu D. Nguyen\thanks{nguyen@rowan.edu}}
\affil[1]{Department of Mathematics, Rowan University}
\date{May 29, 2015}
\begin{document}
\maketitle
\renewcommand\Authands{ and }
\blfootnote{94A05: Communication theory}
\begin{abstract}
We characterize group symmetries of poly-phase complementary code matrices (CCMs), which we use to classify CCMs in terms of their equivalence classes.  We also present classification results for CCMs of dimension $N\times 4$ where $N=2,3,4,5,6$. Finally, we present a new construction to generate quad-phase CCMs from ternary CCMs and compare this to other existing constructions that focus on generating CCMs from those of smaller dimensions. 
\end{abstract}

\section{Introduction}
In phase-coded pulse-compression radar, signal waveforms can be represented by a string of values (called a code) consisting of roots of unity. In order to correlate incoming with outgoing signals these codes should have nice sharp peak-sidelobe characteristics in terms of their auto-correlation functions. Usually in radar and communications, Barker codes and Golay pairs of codes are employed, but there is increased interest in more general codes called complementary code sets, i.e., complementary sets of codes whose composite auto-correlation function has zero sidelobe levels (\cite{ref1,ref4}).  When expressed in matrix form, complementary code sets are referred to as complementary code matrices (CCMs). 

In this paper we characterize known symmetries of CCMs in terms of their relations.  Previously, Golomb and Win \cite{ref7} investigated symmetries of a single polyphase sequence.  These symmetries were later extended to CCMs by Coxson and Haloupek \cite{ref1}.  Let $M$ be a CCM.  Here are the five known symmetries of $M$ that preserves its CCM-property: 
\begin{enumerate}
\item[(i)] Column multiplication by a unimodular complex number.
\item[(ii)] Column conjugate reversal.
\item[(iii)] Matrix conjugation.
\item[(iv)] Progressive multiplication by consecutive powers of a unimodular complex number.
\item[(v)] Column permutation.
\end{enumerate}
When viewed as group generators these five symmetries are non-commutative in general, e.g., column multiplication does not commute with matrix conjugation.  Therefore, it is important to characterize their relations, which we use to classify CCM's in terms of their equivalence classes.

Our results describing the group relations between the five symmetries above extend those of Coxson \cite{ref6} who determined the group structure for symmetries of Barker codes.
Moreover, we obtain an upper bound on the corresponding group generated by these symmetries, which we call the complementary group $G$.  For $p$-phase $N\times K$ CCMs, we establish in Section \ref{sec:3} (Theorem \ref{th:bound-complementary-group}) that
\[
|G|\leq 2^{K+1}p^{K+1}K!.
\]

In Section \ref{sec:4}, we extend Coxson and Russo's \cite{ref2} efficient exhaustive search algorithm for binary CCMs to $p$-phase CCMs.  This algorithm was implemented for quad-phase $N\times 4$ CCMs in Section \ref{sec:5} to obtain a classification of all equivalence classes for $N=2,3,4,5,6$ (see Table \ref{table:equivalence-classes} for a description of the number of equivalence classes and those that are represented by Hadamard matrices). This extends Gibson's \cite{ref3} classification results for quaternary Golay sequence pairs. In the same section, we present a method to construct quad-phase CCMs from dual-pairs of ternary CCMs (dual in the sense that their commutator is diagonally regular).  Lastly, we present in Table \ref{table:construction-methods} a list of the number of equivalence classes whose representatives are dual-pairs.

\section{Preliminaries}
In this section we present definitions of complementary code matrices based on the auto and row-correlation functions. Let $M$ be a $N \times K$ $p$-phase matrix, i.e., one whose entries consists of $p$-th roots of unity. At times we will represent $M$ in several different ways: coordinate-wise with $M=[m_{n,k}]$ $(1 \leq n \leq N,\text{ } 1 \leq k \leq K)$, column-wise with $M= [x_{1}, x_{2}, x_{3},\ldots,x_{K}]$, or row-wise with $M= [r_{1}, r_{2}, r_{3},\ldots,r_{N}]^T$.

  \begin{defn} The aperiodic {\bf autocorrelation function} of a code $x=(a_1,\ldots,a_N)$ is defined to be
  \begin{equation}  \label{eq:acf}
  \mathrm{A}_{x}(j)= 
  \begin{cases}
  \sum\limits_{i=1}^{N-j} a_{i}\bar{a}_{i+j}, & \text{if } 0\leq j \leq N-1; \\
  \\
  \overline{\mathrm{A}_{x}(-j)}, & \text{if } -N+1\leq j <0. \\
\end{cases}
\end{equation}
  \end{defn}
  \noindent Note that $\mathrm{A}_x(0)=|x|^2=N$.  Also, a desirable autocorrelation function should have a sharp center peak ($j=0$) and low sidelobes ($j\neq 0$).

  \begin{exmp}
  Let $x=\{-1,-1,-i\}$.  Then
  \begin{align*}
  \mathrm{A}_{x}(-2) & = -1(i)=-i\\
  \mathrm{A}_{x}(-1) & = -1(-1)-1(i)=1-i \\
  \mathrm{A}_{x}(0) & = -1(-1)-1(-1)-i(-i)=3 \\
  \mathrm{A}_{x}(1) & = -1(-1)-1(-i)=1+i \\
  \mathrm{A}_{x}(2) & = -1(-i)=i.
  \end{align*}
  \end{exmp}

\begin{exmp}
An example of a Barker Sequence, where the $\mid A_x(j)\mid \leq 1$ for all $j\neq 0$, is the following:
\[
x=\{1, 1, 1, 1, 1, -1, -1, 1, 1, -1, 1, -1, 1\}.
\]
The autocorrelation of $x$ is given by
\[
\mathrm{A}_x=\{1, 0, 1, 0, 1, 0, 1, 0, 1, 0, 1, 0, 13, 0, 1, 0, 1, 0, 1, 0, 1, 0, 1, 0, 1\}.
\]
One can observe the peak and sidelobes of this autocorrelation in Figure \ref{fig1}.
\begin{figure}[ht]
\begin{center}
\includegraphics[scale=.7]{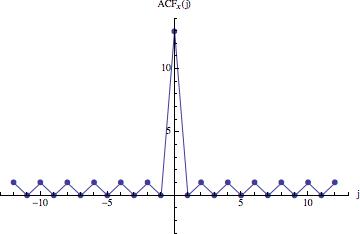}
\caption{Autocorrelation of Barker sequence $x=\{1, 1, 1, 1, 1, -1, -1, 1, 1, -1, 1, -1, 1\}$.}
\label{fig1}
\end{center}
 \end{figure}
\end{exmp}

The goal in radar coding theory is to construct code with zero sidelobes; however, in the case of Barker codes, this is impossible since the calculation of the last sidelobe, $|A_x(N-1)|=|a_1\bar{a_n}|=1$, shows that $A_x$ must always equal 1.  On the other hand, if we employ a set of two complementary codes, i.e., a Golay pair, then it is possible to obtain zero sidelobes.  This is illustrated in the next example.

\begin{exmp}
Consider two codes of length 4, $x_1=\{1,1,-1,1\}$ and $x_2=\{1,1,1,-1\}$. Their autocorrelations are given by $A_{x_1}=\{1,0,-1,4,-1,0,1\}$ and $A_{x_2}=\{-1,0,1,4,1,0,-1\}$.  Then $(x_1,x_2)$ forms a Golay pair since the sum of their autocorrelations has zero sidelobes: 
\[
A_{x_1}+A_{x_2}=\{0,0,0,4,0,0,0\}
\]
\end{exmp}

The notion of a Golary pair can be generalized to a set containing an arbitrary number of complementary codes (see \cite{ref4}).

\begin{defn}
A set of $K$ codes $\{x_{1},\ldots,x_{K}\}$, each of length $N$, is called a \textbf{complementary code set} if
\begin{equation} \label{eq:composite-autocorrelation}
\sum\limits_{i=1}^{K}\mathrm{A}_{x_{i}}(j)=NK\delta_j
\end{equation}
for $-N+1\leq j \leq N-1$ and where $\delta_{j}$ is the Kronecker delta function defined by
\begin{equation*}
\delta_{j}= \begin{cases}
0, & \text{if }j \neq 0; \\
1, & \text{if }j =0.
\end{cases}
\end{equation*}

\end{defn}

Next, we review Coxson and Haloupek's \cite{ref1} matrix formulation of complementary code sets.

\begin{defn}
An $N\times K$ matrix $M$ is a \textbf{complementary code matrix} (CCM) if its columns, $x_{1}, \ldots, x_{N}$, form a complementary code set.  In that case, we define the {\bf composite autocorrelation function} of $M$ by the following
\begin{equation}
\mathrm{A}_{M}(j)=\sum\limits_{i=1}^{K}\mathrm{A}_{x_{i}}(j)
\end{equation}
so that
\begin{equation}
\mathrm{A}_M(j)=NK\delta_j
\end{equation}
because of (\ref{eq:composite-autocorrelation}).
\end{defn}

\begin{defn}
Let $M$ be a matrix.  Then $M\cdot M^*$, where $M^*$ represents the conjugate transpose of $M$, is defined to be the \textbf{row Gramian} of $M$.
\end{defn}
\begin{exmp}
\label{rowgramian}
Let $M=\left(\begin{array}{cccc}1&1&1&1\\1&-1&-1&-1\\1&1&-1&-1\end{array}\right)$.  Then the row Gramian of $M$ is given by
\begin{equation*}
M\cdot M^* = \left(\begin{array}{cccc}1&1&1&1\\1&-1&-1&-1\\1&1&-1&-1\end{array}\right) \cdot \left(\begin{array}{ccc}1&1&1\\1&-1&1\\1&-1&-1\\1&-1&-1\end{array}\right)=\left(\begin{array}{ccc}4&-2&0\\-2&4&2\\0&2&4\end{array}\right)
\end{equation*}
\end{exmp}

\begin{defn}A $N\times N$ matrix $Q$ is \textbf{diagonally regular} if its off-diagonals, i.e., those diagonals outside of the main diagonal, sum to zero. 
\end{defn}

Note that Example 2.3 gives an example of a matrix $M$ whose row Gramian is diagonally regular.  Moreover, $M$ is an example of a complementary code set.  This not a coincidence.

\begin{lem}[Coxson-Haloupek \cite{ref2}]
An $N\times K$ matrix $M$ is a CCM if its row Gramian is diagonally regular, i.e.,
\begin{align*}
M\cdot M^* = Q,
\end{align*}
where Q is diagonally regular.
\end{lem}

Next, we present the row-correlation function, which gives an equivalent representation of a CCM in terms of its rows \cite{ref2,ref5}.

\begin{defn} The $\textbf{row-correlation function}$ of a matrix $M$ consisting of rows $\{r_1,\ldots,r_K\}$ is defined by
\begin{align*}
\mathrm{R}_{M}(j)=\sum\limits_{i=1}^{N-j}r_{i} \cdot \bar{r}_{i+j}
\end{align*}
for $0 \leq j \leq N-1$ and $\mathrm{R}_{M}(j)=\overline{\mathrm{R}_{M}(-j)}$ for $-N+1 \leq j <0$.
\end{defn}
Note that $R_{M}(j)$ represents the sum of the elements in the $j$-th diagonal of $M\cdot M^*$, where $j=0$ corresponds to the main diagonal and $j>0$ (or $j<0$) corresponds to $j$-th diagonal above (or below) the main diagonal, respectively.  Moreover, these row-correlation functions can also be represented as the sum of the autocorrelation function of each column in $M$.
\begin{exmp}
The row-correlation function of $M$ given in Example \ref{rowgramian} is given by the following.
\begin{align*}
\mathrm{R}_{M}(0) & = r_{1}\cdot \bar{r}_{1}+r_{2}\cdot \bar{r}_{2}+r_{3}\cdot \bar{r}_{3} \\
& = [-1,-1,-i]\cdot  [-1,-1,i]+[-i,-i,1]\cdot  [i,i,1]+[-1,-1,1]\cdot  [-1,-1,1]= [3,3,3] \\
\mathrm{R}_{M}(1) & = r_{1}\cdot \bar{r}_{2}+r_{2}\cdot \bar{r}_{3}= [-1,-1,-i]\cdot  [i,i,1]+[-i,-i,1]\cdot  [-1,-1,1]= [-3i,1+2i] \\
\mathrm{R}_{M}(2) & = r_{1}\cdot \bar{r}_{3}= [-1,-1,-i]\cdot  [-1,-1,1]= [2-i]\\
\mathrm{R}_{M}(-1) & = \overline{\mathrm{R}_{M}(1)} = [3i, 1-2i]\\
\mathrm{R}_{M}(-2) & = \overline{\mathrm{R}_{M}(2)} = [2+i]
\end{align*}
\end{exmp}

The following theorem characterizes CCMs in terms of the row-correlation function.
\begin{thm}[\cite{ref2,ref5}]
A $N\times K$ matrix $M$ is a CCM if and only if
\begin{equation*}
\mathrm{R}_{M}(j)= \sum_{i=1}^{K}A_{x_{i}}(j)=NK \delta_{j}
\end{equation*}
for $-N+1 \leq j \leq N-1$.
\end{thm}

\section{Symmetries of CCMs}
\label{sec:3}

Let $C_{N,K}(p)$ denote the set of all CCMs of dimension $N\times K$ consisting of $p$-th roots of unity.  Any function $f:C_{N,K}(p) \rightarrow C_{N,K}(p)$ that maps CCMs to CCMs of the same dimension will be called a CCM-preserving symmetry of $C_{N,K}(p)$.  In particular, if $A_M(j)$ represents the composite autocorrelation of $M$, then we require
\[
A_{f(M)}(j)=0
\]
for all $M\in C_{N,K}(p)$ and $|j|=0,1,\ldots,N-1$.

Coxson and Haloupek discussed the different ways to create a new $N \times K$ CCM using its symmetries as described in the following theorem \cite{ref1}.
\begin{thm}Suppose that $M$ is a $p$-phase $N\times K$ CCM.  Then the following operations are CCM-preserving symmetries, i.e., the resulting matrix is a $N\times K$ CCM.
\begin{itemize}
\item[(i)] Take any column $x$ of $M$ and replace, $x$, with $\alpha x$ where $\alpha$ is any unit-modulus complex number.
\item[(ii)] Take any column $x=[m_{1}, \ldots , m_{N}]$ of $M$ and replace $x$ with its conjugate reversal $\hat{x}=[\bar{m}_{N}, \ldots, \bar{m}_{1}]$.
\item[(iii)] Form $\bar{M}$.
\item[(iv)] Form the product $Diag[{\beta, \beta^2, \beta^3,\ldots,\beta^N}]M$, with $\beta$ being any unit-modulus complex number.  
\item[(v)] Form $MP$ where $P$ any $K\times K$ permutation matrix.
\end{itemize}
\label{constructions}
\end{thm}

Using this theorem we can see that with one CCM many more can be created. Coxson and Russo created an algorithm for narrowing down the search space for complementary code matrices using these symmetries \cite{ref2} .  This is important due to the fact that the search space for p-phase $N\times K$ CCMs is $p^{NK}$; as the dimensions increase the search takes exponentially longer to conduct an exhaustive search. For example, the exhaustive search for 4-phase $4 \times 4$ CCMs, which generated $4300800$, took about 24 hours using parallel processing on 4 cpus.  With the help of this theorem, it is possible to narrow down the search space if it is possible to advoid checking all but one from a set of equivalent CCMs. There is a wonderful example of this by Gibson where he classified 4-phase Golay pairs into equivalence classes \cite{ref3}.  We aim to do the same, but for more general CCMs. 

A quick example of this is as follows. 
\begin{exmp}
You can see here that we can take one matrix $M$ and make two other matrices out of $M$, though we can make many more besides these two using Theorem \ref{constructions}.  To get the first matrix we applied (i) multiplying $x_{1}, x_{2}, $ and $x_{4}$ by $-i$ and $x_{3}$ by $-1$. Then to that matrix we applied (iv) with $\beta = -i$. Yet the other interesting thing about this example is that we were able to make the entire first row ones as well as $m_{2,1}$.  This is what we call a normalized CCM.

\begin{equation}
M=\left(\begin{array}{cccc}-1&-1&-i&-1\\-i&-i&1&1\\-1&-1&1&1\\i&-i&-1&-i\end{array}\right)
\end{equation}
$M_{1}=[-ix_{1}, -ix_{2}, -x_{3}, -ix_{4}]$
\begin{equation}
M_{1}=\left(\begin{array}{cccc}i&i&i&i\\-1&-1&-1&-i\\i&i&-1&-i\\1&-1&1&-1\end{array}\right)
\end{equation}
$M_{2}=Diag[-i, 1, i, -1]M_{1}$
\begin{equation}
\label{eq:1}
M_{2}=\left(\begin{array}{cccc}1 & 1 & 1 & 1 \\1 & 1 & 1 & i\\ -1 & -1 & -i & 1\\ 1&-1&1&-1 \end{array}\right)
\end{equation}
\end{exmp}

\subsection{Complementary Codes}
We proceeded to find the relationships between different CCM symmetries.  Our results generalize Coxson's who examined symmetries that preserve low sidelobe levels in Barker codes \cite{ref6}.  We begin with definitions and notations to express the symmetries mentioned in Theorem \ref{constructions}.
\begin{defn} Let the following assist in representing the different operations Theorem \ref{constructions}.
\begin{itemize}
\item[(a)]$\mathcal{Z} = \{(\alpha_{1}, \alpha_{2}, \ldots , \alpha_{K}) : \alpha_{k}^p=1 \}$. 
\item[(b)]$\mathcal{A}=\{(t_{1}, t_{2}, t_{3}, \ldots , t_{K}) : t_{k} \in \{0,1\}\}$.
\end{itemize}
\end{defn}
\begin{defn} The following are definitions that represent the CCM preserving operations, Theorem \ref{constructions}, where $U \in \mathcal{Z}$ and $T \in \mathcal{A}$.
\begin{itemize}
\item[(a)]$C$ denotes multiplying columns in $M$ Theorem \ref{constructions} (i)
\begin{equation*}
C_{U}M=[\alpha_{1}x_{1}, \alpha_{2}x_{2},\ldots, \alpha_{K}x_{K}] 
\end{equation*}
\item[(b)]$\rho$ denote the conjugate reversals of columns in $M$. Theorem \ref{constructions} (ii) 
\begin{equation*}
\rho_T M=[\omega^{2-t_{1}}(x_{1}), \omega^{2-t_{2}}(x_{2}),\ldots, \omega^{2-t_{K}}(x_{K})] 
\end{equation*}
\item[(c)]$S$ denote the conjugate of the matrix $M$. Theorem \ref{constructions} (iv)
\begin{equation*}
SM=\bar{M}
\end{equation*}
\item[(d)]$Q(\beta)= N\times N$ diagonal matrix with values going from $\beta, \beta^2,\ldots\beta^N$ Theorem \ref{constructions} (iv)
\begin{equation*}
Q(\beta) M=[\beta r_{1}, \beta^2 r_{2},\ldots, \beta^N r_{N}]= [Q(\beta)x_{1},Q(\beta)x_{2},\ldots,Q(\beta)x_{K}] 
\end{equation*}
\item[(e)]$P=$ a permutation of the columns. Theorem \ref{constructions} (v)
\begin{equation*}
P M= [x_{\sigma(1)}, x_{\sigma(2)}, x_{\sigma(3)},\ldots,x_{\sigma(K)}]
\end{equation*}
\item[(f)]$\bar{U}=[\bar{\alpha_{1}}, \bar{\alpha_{2}}, \bar{\alpha_{3}}, \ldots , \bar{\alpha_{K}}]$
\item[(g)]$U_{P}=[\alpha_{\sigma(1)}, \alpha_{\sigma(2)}, \alpha_{\sigma(3)}, \ldots , \alpha_{\sigma(K)}]$
\item[(h)]$U_{T,\beta}=[\alpha_{1}, \alpha_{2}, \alpha_{3}, \ldots , \alpha_{K}]$ such that where $t_{k} \in T$
\begin{displaymath}
   \alpha_{k} = \left\{
     \begin{array}{lr}
      \beta^{N+1} & t_{k}=1 \\
    1 & t_{k}=0\\
     \end{array}
   \right.
\end{displaymath} 
\item[(i)] $T_{P^{-1}}= [t_{\sigma^{-1}(1)}, t_{\sigma^{-1}(2)}, t_{\sigma^{-1}(3)}, \ldots, t_{\sigma^{-1}(K)}]$
\item[(j)] $U_T=[\omega^{2-t_1}(\alpha_1),\ldots, \omega^{2-t_K}(\alpha_K)]$
\end{itemize}
\end{defn}
We now establish relationships between symmetries by viewing them as group generators.
\begin{lem} \label{lem:relations} Let $M$ be a $N\times K$ p-phase CCM, $M=[x_{1}, x_{2}, x_{3},\ldots,x_{K}]$ and \\$x_{k}=[m_{1,k}, m_{2, k}, m_{3,k}, \ldots , m_{N,k}]^T$. Then
\begin{itemize}
\item[(i)] $C_{U}\rho_{T}M = \rho_{T}C_{U_{T}}M$ 
\item[(ii)] $C_{U}SM = S C_{\bar{U}}M$ 
\item[(iii)] $C_{U}Q(\beta)M=Q(\beta)C_{U}M$ 
\item[(iv)] $C_{U}PM=PC_{U_{P}}M$ 
\item[(v)] $\rho_{T}SM=S \rho_{T}M$
\item[(vi)] $\rho_{T}PM=P\rho_{T_{P^{-1}}}M$
\item[(vii)] $S Q(\beta)M= Q(\bar{\beta})S M$
\item[(viii)] $S PM=P SM$
\item[(ix)] $Q(\beta) \rho_{T} M=C_{U_{T, \beta}}\rho_{T} Q_{\beta}M$
\item[(x)] $Q(\beta) P M= P Q(\beta) M$
\end{itemize}
\end{lem}

\begin{proof}[proof] 
\begin{itemize}
\item[(i)]\begin{displaymath}
\begin{array}{rl}
 C_{U}\rho_{T}M &=C_{U}\rho_{T}[x_{1}, x_{2}, x_{3}, \ldots , x_{K}]\\ 
&= C_{U}[\omega^{2-t_{1}}(x_{1}), \omega^{2-t_{2}}(x_{2}), \omega^{2-t_{3}}(x_{3}), \ldots , \omega^{2-t_{K}}(x_{K})]\\
&=[\alpha_{1}\omega^{2-t_{1}}(x_{1}), \alpha_{2}\omega^{2-t_{2}}(x_{2}), \alpha_{3}\omega^{2-t_{3}}(x_{3}), \ldots , \alpha_{K}\omega^{2-t_{K}}(x_{K})]\\
&=\rho_{T}[\omega^{2-t1}(\alpha_{1})x_{1}, \omega^{2-t_{2}}(\alpha_{2})x_{2}, \omega^{2-t_{3}}(\alpha_{3})x_{3}, \ldots , \omega^{2-t_{K}}(\alpha_{K})x_{K}]\\
&=\rho_{T}C_{U_{T}}[x_{1}, x_{2}, x_{3}, \ldots , x_{K}]\\
&=\rho_{T}C_{U_{T}}M
\end{array}
\end{displaymath}
\item[(v)]\begin{displaymath}
\begin{array}{rl}
\rho_{T}SM &=\rho_{T}S[x_{1}, x_{2}, x_{3}, \ldots , x_{K}]\\ &= \rho_{T}[\bar{x}_{1}, \bar{x}_{2}, \bar{x}_{3}, \ldots , \bar{x}_{K}]\\
&=[\omega^{2-t_{1}}(\bar{x}_{1}), \omega^{2-t_{2}}(\bar{x}_{2}), \omega^{2-t_{3}}(\bar{x}_{3}), \ldots , \omega^{2-t_{K}}(\bar{x}_{K})]\\
&=S[\omega^{2-t_{1}}(x_{1}), \omega^{2-t_{2}}(x_{2}), \omega^{2-t_{3}}(x_{3}), \ldots , \omega^{2-t_{K}}(x_{K})]\\
&=S \rho_{T}[x_{1}, x_{2}, x_{3}, \ldots , x_{K}]\\
&=S \rho_{T}M
\end{array}
\end{displaymath}
\item[(vi)] \begin{displaymath}
\begin{array}{rl}
\rho_{T}PM &=\rho_{T}P[x_{1}, x_{2}, x_{3}, \ldots , x_{K}]\\ 
&= \rho_{T}[x_{\sigma(1)}, x_{\sigma(2)}, x_{\sigma(3)}, \ldots , x_{\sigma(K)}]\\
&=[\omega^{2-t_{1}}(x_{\sigma(1)}), \omega^{2-t_{2}}(x_{\sigma(2)}), \omega^{2-t_{3}}(x_{\sigma(3)}), \ldots , \omega^{2-t_{K}}(x_{\sigma(K)})]\\
&=P[\omega^{2-t_{\sigma^{-1}(1)}}(x_{1}), \omega^{2-t_{\sigma^{-1}(2)}}(x_{2}), \omega^{2-t_{\sigma^{-1}(3)}}(x_{3}), \ldots , \omega^{2-t_{\sigma^{-1}(K)}}(x_{K})]\\
&=P\rho_{T_{P^{-1}}}[x_{1}, x_{2}, x_{3}, \ldots , x_{K}]\\
&=P\rho_{T_{P^{-1}}}M
\end{array}
\end{displaymath}
\item[(ix)] \begin{displaymath}
\begin{array}{rl}
Q(\beta) \rho_{T} $M$ &= Q(\beta)\rho_{T}[x_{1}, x_{2}, x_{3},\ldots,x_{K}]\\
&=Q(\beta)[\omega^{2-t_{1}}(x_{1}),\omega^{2-t_{2}}(x_{2}),\omega^{2-t_{3}}(x_{3}),\ldots,\omega^{2-t_{K}}(x_{K})]\\
&=[Q(\beta)\omega^{2-t_{1}}(x_{1}),Q(\beta)\omega^{2-t_{2}}(x_{2}),Q(\beta)\omega^{2-t_{3}}(x_{3}),\ldots,Q(\beta)\omega^{2-t_{K}}(x_{K})]\\
&=C_{U_{T,\beta}}[\bar{\alpha}_{1}Q(\beta)\omega^{2-t_{1}}(x_{1}),\bar{\alpha}_{2}Q(\beta)\omega^{2-t_{2}}(x_{2}),\bar{\alpha}_{3}Q(\beta)\omega^{2-t_{3}}(x_{3}),\ldots,\bar{\alpha}_{K}Q(\beta)\omega^{2-t_{K}}(x_{K})]\\
&=C_{U_{T,\beta}}\rho_{T}[\omega^{2-t_{1}}(\bar{\alpha}_{1}Q(\beta)\omega^{2-t_{1}}(x_{1})),\omega^{2-t_{2}}(\bar{\alpha}_{2}Q(\beta)\omega^{2-t_{2}}(x_{2})),\ldots,\omega^{2-t_{K}}(\bar{\alpha}_{K}Q(\beta)\omega^{2-t_{K}}(x_{K}))]\\
&=C_{U_{T,\beta}}\rho_{T}[\omega^{2-t_{1}}(\bar{\alpha}_{1}Q(\beta))x_{1},\omega^{2-t_{2}}(\bar{\alpha}_{2}Q(\beta))x_{2},\ldots,\omega^{2-t_{K}}(\bar{\alpha}_{K}Q(\beta))x_{K}]\\
&=C_{U_{T,\beta}}\rho_{T}Q(\beta)[x_{1}, x_{2}, x_{3}, \ldots, x_{K}]\\
&= C_{U_{T,\beta}}\rho_{T} Q_{\beta}M
  \end{array}
  \end{displaymath}
  \end{itemize}
 The proofs for parts (ii), (iii), (iv), (vii), (viii), and (x) are analogous and will be omitted.
\end{proof}

\begin{defn}
The complementary group $G$ of the set of all $N\times K$ $p$-phase CCMs is defined to be the group generated by the symmetries $S,P,C_{U},\rho_{T}, Q(\beta)$ and their relations given in Lemma \ref{lem:relations}.
\end{defn}

\begin{thm} \label{th:bound-complementary-group}
The cardinality of the complementary group $G$ of $C_{N,K}(p)$ is bounded by 
\begin{equation}
|G|\leq 2^{K+1}p^{K+1}K!
\end{equation}
\end{thm}

\begin{proof}
Each CCM preserving operation on the matrix $M$ can be represented in the following form.
\begin{equation*}
S P C_{u} \rho_{T} Q(\beta)
\end{equation*}
Since $S$ represents the conjugate of the matrix there are two possibilities, i.e., $|S|=2$.  $P$ represents permutations of the columns $|P|=K!$.  $C_{U}$ is multiplying columns by different unit-modulus complex number i.e.  $|C_{U}|=p^K$.  Whether or not the conjugate reversal of a column is take has $|\rho_{T}|=2^K$.  And if we do progressive multiplication there is $|Q(\beta)|=p$.
So this will will produce a max of 
\begin{equation*}
|G|\leq |S||P||C_{U}||\rho_{T}||Q(\beta)|=2 K! p^K 2^K p= 2^{K+1}p^{K+1}K!
\end{equation*} CCMs from $M$.   
\end{proof}

\section{CCM Search Algorithm}
\label{sec:4}

The following is an explanation of our exhaustive search algorithm for CCMs.  Our search algorithm is based on the search algorithm of Coxson and Russo \cite{ref2}.  While their work searches for the binary CCMs, ours focuses on any poly-phase. We begin by presenting how a CCM can be represented in a normalized form as shown in the following lemma an example of which was expressed in Example \ref{eq:1}.

\begin{lem}[\textbf{Normalized CCM}]
Any $N\times K $ CCM can be normalized by transforming the matrix into another matrix (using Theorem \ref{constructions}) where the entire first row consists of 1s and the first element in the second row is a 1.
\end{lem}
\begin{proof}
Let $M$ be a $N\times K $ CCM with $r_{1}=[m_{1,1}, m_{1,2}, m_{1,3}, \ldots , m_{1,K}]$ and  $r_{2}=[m_{2,1}, m_{2,2}, m_{2,3}, \ldots , m_{2,K}]$.  In this proof it is sufficient to just show the operations being applied to the first two rows, but the reader may note that these operations will be applied to the other rows as well.

Here we used (i) in Theorem \ref{constructions} multiplying each column of $M$ by the conjugate of the first value of each column.  This makes the entire first row ones since $x\bar{x}=1$.
\begin{align*}
\left(\begin{array}{ccccc}m_{1,1}&m_{1,2}&m_{1,3}&\ldots&m_{1,K}\\m_{2,1}&m_{2,2}&m_{2,3}&\ldots&m_{2,K}\\
\vdots&\vdots&\vdots&\vdots&\vdots \end{array}\right) \Rightarrow \left(\begin{array}{ccccc}\bar{m}_{1,1} m_{1,1}&\bar{m}_{1,2}m_{1,2}&\bar{m}_{1,3}m_{1,3}&\ldots&\bar{m}_{1,K}m_{1,K}\\ \bar{m}_{1,1} m_{2,1}&\bar{m}_{1,2}m_{2,2}&\bar{m}_{1,3}m_{2,3}&\ldots&\bar{m}_{1,K}m_{2,K}\\
\vdots&\vdots&\vdots&\vdots&\vdots \end{array}\right)
\end{align*}
Next we uses (iv) in Theorem \ref{constructions} and the fact that Diag$(\overline{m_{2,1} \bar{m}_{1,1}})=$Diag$(\bar{m}_{2,1} m_{1,1})$.  
\begin{align*}
\left(\begin{array}{ccccc}\bar{m}_{2,1} m_{1,1}&\bar{m}_{2,1} m_{1,1}&\bar{m}_{2,1} m_{1,1}&\ldots&\bar{m}_{2,1} m_{1,1}\\ (\bar{m}_{2,1}m_{1,1})^2 \bar{m}_{1,1} m_{2,1}& (\bar{m}_{2,1}m_{1,1})^2 \bar{m}_{1,2}m_{2,2}& (\bar{m}_{2,1}m_{1,1})^2 \bar{m}_{1,3}m_{2,3}&\ldots& (\bar{m}_{2,1}m_{1,1})^2 \bar{m}_{1,K}m_{2,K}\\
\vdots&\vdots&\vdots&\vdots&\vdots \end{array}\right)
\end{align*}
The final step is to now multiply each column by $\overline{\bar{m}_{2,1} m_{1,1}}=m_{2,1} \bar{m}_{1,1}$ to bring the first row back to all 1's.  Note now though that in this new matrix, $M'$, $m'_{2,1}=(\bar{m}_{2,1}m_{1,1})^2 (\bar{m}_{1,1} m_{2,1})^2$ also known as $m'_{2,1}=1$  The normalized version of $M$ can now be written as $M'$.
\begin{align*}
M'=\left(\begin{array}{ccccc}1&1&1&\ldots&1\\
1&m'_{2,2}&m'_{2,3}&\ldots&m'_{2,K}\\
\vdots&\vdots&\vdots&\vdots&\vdots \end{array}\right)
\end{align*}
\end{proof}

\subsection{Short summary of search algorithm}

The algorithm is an exhaustive search for all representations of CCMs.  The search implements the known symmetries to avoid finding a mass amount of equivalent CCMs. We construct the matrix from the outside in.  We start with making sure the matrix is normalized, setting the first row all equal to ones and first element of second row a one.  We then force the last row to sum to zero.  Next we imagine the unit circle and name the first angle, between the origin and the first of the p-phase, $\omega$. We can define a method of sorting where $0\omega<\omega<2\omega<\ldots<(p-1)\omega$. When constructing $r_{2}$ we want the corresponding $K$-tuples to follow this sorting process.  We then construct $r_{N-1}$ by only using $K$-tuples whose sum is equal to the opposite of the partial sum of the row-correlation function of the second to last off diagonal, meaning summing up all of the off diagonal minus the one that includes the row that we are on.  We save time here by not looking at $K$-tuples that we know will make the $N-2$ off diagonal not sum to zero.  From here we create $r_{3}$ based off of $r_{2}$.  If $r_{2}$ has consecutive values that are equal, then at those positions in $r_{3}$ the values need to be sorted.  Finally we continue searching for values for the other rows by continuously looking at their partial sum of their row-correlation function.  Once we have created all of these rows we just need to check that the remaining $N-1-\frac{N}{2}$ off diagonals also sum to zero.  If so the matrix is a CCM.

There are many similarities and one major difference between our search algorithm and Jon Russo's.  The similarities consist of constructing the code from the outside in and calculating the autocorrelation functions as the search progresses.  Both algorithms normalize the first row, to account for Theorem \ref{constructions} (i), and set the sum of the last row to zero. They also make sure the conjugate reversal of the column is less than the column, accounting for Theorem \ref{constructions} (ii). Both also make sure that the columns are ordered to avoid permuting the rows. The difference between the search algorithms is that ours is generalized for all $N\times K$ poly-phase CCMs.

One of the most difficult parts of this algorithm to understand is the sorting by increasing angles.  For this reason we will begin with an example for the reader.  Another way one might think of the increasing angles is by increasing exponents.  In a 4-phase CCM we deal with $\{1, i, -1, -i\}$ or $\{i^0, i^1, i^2, i^3\}$. 

\begin{defn}
Let $r_{n}$ be the $n$-th row of a matrix $M$.

\begin{exmp}
Let $r_{n}=[1,1,i]$
Then the options for the increasing exponents of $r_n$ are as follows:  $\{[1,1,1],$ $[1,i,1],$ $[1,-1,1],$ $[1,-i,1],$ $[i,i,1],$ $[i,-1,1],$ $[i,-i,1],$ $[-1,-1,1],$ $[-1,-i,1],$ $[-i,-i,1],$ $[1,1,i],$ $[1,i,i],$ $[1,-1,i],$ $[1,-i,i],$ $[i,i,i],$ $[i,-1,i],$ $[i,-i,i],$ $[-1,-1,i],$ $[-1,-i,i],$ $[-i,-i,i],$ $[1,1,-1],$ $[1,i,-1],$ $[1,-1,-1],$ $[1,-i,-1],$ $[i,i,-1],$ $[i,-1,-1],$ $[i,-i,-1],$ $[-1,-1,-1],$ $[-1,-i,-1],$ $[-i,-i,-1],$ $[1,1,-i],$  \\$[1,i,-i],$ $[1,-1,-i],$ $[1,-i,-i]$ $[i,i,-i],$ $[i,-1,-i],$ $[i,-i,-i],$ $[-1,-1,-i],$ $[-1,-i,-i],$ $[-i,-i,-i]\}$
\end{exmp}
Instead of there being $4^3$ possibilities for the next row there are $3\times 4^2$.  As the rows get longer or $r_{n}$ changes, this can save more and more time. All we are doing here is sorting the columns as Theorem \ref{constructions} (v) allows.

The beauty of constructing the matrix from the outside in is that this construction method allows for the next key part of this algorithm.  We can sum up the elements off diagonals of the partially constructed matrices dot product with its conjugate transpose. If we sum up all of the terms except the term that is multiplied to row one.  This sum will first have to be less than or equal to K by in order for this to happen.  
\begin{equation*}
\left(
\begin{array}{c}
r_1 \\
r_2 \\
r_3 \\
r_4
\end{array}
\right)\cdot
\left(
\begin{array}{cccc}
 \bar{r}_1 & \bar{r}_2 & \bar{r}_3 & \bar{r}_4
\end{array}
\right)=
\left(
\begin{array}{cccc}
 r_1 \bar{r}_1 & r_1 \bar{r}_2 & r_1 \bar{r}_3 & r_1 \bar{r}_4 \\
 r_2 \bar{r}_1 & r_2 \bar{r}_2 & r_2 \bar{r}_3 & r_2 \bar{r}_4 \\
 r_3 \bar{r}_1 & r_3 \bar{r}_2 & r_3 \bar{r}_3 & r_3 \bar{r}_4 \\
 r_4 \bar{r}_1 & r_4 \bar{r}_2 & r_4 \bar{r}_3 & r_4 \bar{r}_4 
\end{array}
\right)
\end{equation*}

\begin{itemize}
\item[(a)]Increasing exponents (E$(r_{n})$): 
creates a list of new rows such that for each row $r_{n^{'}}$,  $r_{n^{'}}= [i^{p_{1}},\ldots,i^{p_{K}}]$ such that $r_{n^{'}}$ is sorted based off of $r_{n}$.  Represent $r_{n}=[m_{n,1},m_{n,2},\ldots,m_{n,K}]$ for any consecutive equal m's in $r_{n}$ the $r_{n^{'}}$ of those values must be in increasing exponential form mod 4. If $r_{n}$ happens to be normalized then E$(r_{n})=$ a list of $r_{n^{'}}$'s where $r_{n^{'}}=[i^{p_{1}},i^{p_{2}},\ldots,i^{p_{K}}]$ and $p_{k} $ mod 4 $\leq p_{k+1}$ mod 4. If $r_{n}=[1,1,i,-i]=[i^0, i^0, i^1, i^3 ]$ then  E$(r_{n})=$ a list of $r_{n^{'}}$'s where $r_{n^{'}}=[i^{p_{1}},i^{p_{2}},\ldots,i^{p_{K}}]$  where $p_{1}$ mod 4 $\leq p_{2}$ mod $4$ and $0\leq p_{3} \leq 3$, $0\leq p_{4} \leq 3$.
\item[(b)]Define the absolute value of a complex number as the Taxicab Distance: $\| (a+bi) \| = |a|+|b|$ 
\end{itemize}
\end{defn}
\begin{algorithm}
\caption{Search Algorithm for $N\times K$ CCMs}
\begin{algorithmic}
\STATE Normalize $r_{1}$
\COMMENT{from this point forward we will write $\bar{r}_{1}$ as $r_{1}$}
\FORALL {K-tuples that sum to zero}
    \STATE set $r_{N}$ equal to this K-tuple
    \IF{$N=2$}
        \STATE This is a CCM
    \ENDIF
    \STATE Set $m_{2,1}=1$    
    \COMMENT{now the matrix is fully normalized}
    \FORALL {K-tuples in E($r_{1}$) }
        \STATE set $r_{2}$ equal to this K-tuple;
        \COMMENT{Partial sum of the rows 2}
        \IF{$N=3$}
            \IF{$R_M(1)=0$}
                \STATE This is a CCM
            \ELSE
                \STATE NOT A CCM
            \ENDIF
        \ENDIF
            \IF {$\|$R$_{M}(2-N)-r_{[N-1]}r_{1}\| \leq K$}
                \FORALL {K-tuples that are equivalent to $-($R$_{M}(2-N)-r_{[N-1]}r_{1})$}
                    \STATE set $r_{N-1}$ equal to this K-tuple
                    \WHILE {$2 \leq k \leq K$, $m_{2,k}=i^{P}$ and $\bar{m}_{N-1,k} m_{N,k}= i^{P'}$ where $0\leq P \leq 3$ and  $0\leq P' \leq 3$}
                        \IF {$P>P'$}
                            \STATE Not the CCM we are looking for;
                            \COMMENT{This step is taking into account the conjugate reversals of a column}
                        \ENDIF
                    \ENDWHILE
                    \FORALL {K-tuples in E($r_{2}$)}
                        \STATE set $r_{3}$ equal to this K-tuple
                            \FOR {$t=1$, $t\leq \frac{N}{2}-2 $, $t++$}
                                \IF {$\|$R$_{M}(t+2-N)-r_{[N-t-1]}r_{1}\| \leq K$}
                                    \FORALL {K-tuples whose sums are equivalent to $-($R$_{M}(t+2-N)-r_{[N-t-1]}r_{1})$}
                                        \STATE set $r_{[N-t-1]}$ equal to this K-tuple
                                    \ENDFOR
                                \ENDIF
                            \ENDFOR
                            \FOR {$s=1$, $s\leq N-1-\frac{N}{2}$, $s++$}
                                \IF {R$_{M}(s) \neq 0$}
                                    \STATE NOT A CCM; \COMMENT{Check that the remaining off diagonals sum to zero}
                                \ENDIF 
                            \ENDFOR
                            \STATE This is a CCM
                    \ENDFOR
                \ENDFOR
            \ENDIF 
    \ENDFOR
\ENDFOR
\end{algorithmic}
\end{algorithm}

\subsection{The Search}
This search was implemented in C++ using the adaptation of the Coxson-Russo search algorithm mentioned in this section.  Since we used the symmetries to speed up the search, we did not exhautively apply the symmetries.  So once the search was completed, a Mathematica program was implemented on the matrices found where all combinations of symmetries being applied to them.  The purpose was to narrow it down to the exact number of equivalence classes.  The results of which are shown in Table \ref{table:equivalence-classes}.

\section{Equivalence Classes of CCMs}
\label{sec:5}

\begin{defn}[\textbf{Equivalent CCMs}] Let $M_{1}, M_{2} \in C_{N,K}(p)$ and $g \in G$, where $G$ is the complementary group. Then $M_1$ and $M_2$ are said to be equivalent if 
\begin{equation*}
g(M_{1})=M_{2}
\end{equation*}
\end{defn}
We define $\mathcal{C}_M$ to be the equivalence class of $C_{N,K}(p)$ containing $M$, i.e.,
\[
\mathcal{C}_M = \{\tilde{M}\in C_{N,K}(p): gM=\tilde{M}, g\in G\}
\]
The number of $N \times 2$ $4$-phase CCMs for $N=1,2,3,\ldots,22$ are already known as previously discussed \cite{ref3}. So we will look at what this definition means for $N \times 4$ CCMs. The order of the symmetry group of $4$-phase $N \times 4$ CCMs $786432$.  We must check each of the resulting matrices to ensure that the set of CCMs is narrowed down to the simplest number of equivalence classes. With this theorem we able to narrow down the $4$-phase $4 \times 4$ CCMs to 24 and $4$-phase $6 \times 4$ CCMs to 1448 with the assurance that these are the number of equivalence classes, no more, no less.  We also found the number of equivalence classes for other, easier $N$'s.  We then looked at the equivalence classes to see how many of them could be formed using a Hadamard matrix.  It is interesting to mention here that the entire equivalence class need not be Hadamard matrix in order to have a Hadamard matrix representation.  This is because taking the conjugate reversal of a column does not always preserve the Hadmard matrix property

\begin{table}[H]
\caption{Equivalence Classes}
\label{table:equivalence-classes}
\centering
\begin{tabular}{c|c|c|c}
\hline\hline
$N\times K$ p-phase CCMs & *CR Algorithm & Number of Equivalency Classes & Hadamard Representations\\
\hline\\[-1.0em]
2x4 4-phase & 36 &2&2\\
3x4 4-phase & 95 &5&5\\
4x4 4-phase & 231 & 24& 17 \\
5x4 4-phase & 5246 &133&0\\
6x4 4-phase & 23448 &1448& 0\\
\hline\hline
\end{tabular}
\end{table}
*Coxson and Russo's adapted search algorithm

\subsection{Constructing Quaternary CCMs }
When dealing with CCMs, many times it is useful to look at constructions that can be used to create larger CCMs from smaller ones since the search space for the smaller CCM is, of course, smaller.  It is here that the reader could think back to some of the more well known constructions.  The Kronecker product, represented by the symbol $\oplus$, can take matrix $M_{1}$ that is a $N_{1} \times K_{1}$ CCM and $M_{2}$ that is a $N_{2} \times K_{2}$ CCM to create a new matrix $M$ where $M_{1} \oplus M_{2}=N_{1}N_{2}\times K_{1}K_{2}$.  In the concatenation theorem you can take $M_{1}$ to be a $N \times K_{1}$ CCM and $M_{2}$ to be a $N \times K_{2}$ CCM.  The concatenation theorem gives us the following: $[M_{1}, M_{2}]=N \times (K_{1}+K_{2})$ \cite{ref1}.  There are even more that you can read about yet all the ones from Coxson and Haloupek focus, as we said, on creating a CCM from a smaller dimension CCM to decrease the search space \cite{ref1}.  We can also look at creating a quad-phase CCM from the ternary CCMs.  These are a looser definition of CCMs whose entries consist of the values $\{-1,0,1\}$ and their off diagonals sum to zero. Here we decrease the search space, $3^{NK}$ as opposed to $4^{NK}$.
Say we have a matrix $M$ that is a $4$-phase CCM.
\begin{equation*}
M=\left(\begin{array}{cccc}-1 & -1 & -i & -i \\-i & -i & i & i\\ i & -i & -1 & 1\\ i&-i&i&-i \end{array}\right)
\end{equation*}
We can decompose $M$ into the matrix's real and imaginary parts $M=A+iB$, where $A$ and $B$ are ternary CCM's.
\begin{align*}
A=\left(\begin{array}{cccc}-1 & -1 & 0 & 0 \\0 & 0 & 0 & 0\\ 0 & 0 & -1 & 1\\ 0&0&0&0 \end{array}\right), 
B=\left(\begin{array}{cccc}0 & 0 & -1 & -1 \\-1 & -1 & 1 & 1\\ 1 & -1 & 0 & 0\\ 1&-1&1&-1 \end{array}\right) 
\end{align*}

Though using the ternary CCMs may seem trivial, they can have power. 

\begin{thm}[\textbf{Dual Pair Theorem}]
Assume that $A$ and $B$ are $NK$ ternary CCMs. Then $Z=A+iB$ is a $N\times K$ quad-phase CCM if\\
(i) $\mid A_{n,k} \mid + \mid B_{n,k} \mid =1$ $\{\forall n,k  | 1 \leq n \leq N$ and $ 1\leq k \leq K\}$\\
(ii) $BA^* -B^*A$ is Diagonally regular 
\end{thm}
NOTE: Any pair of matrices $(A,B)$ that satisfy condition (i) is called a dual pair.

\begin{proof}
We first prove that $Z$ is quad-phase by considering two cases: \\
Case 1:$A_{i,j}=0$ and $B_{i,j} = \pm 1$.  Then $Z_{i,j}= A_{i,j} +iB_{i,j} \Rightarrow Z_{i,j} = 0 \pm i \Rightarrow Z_{i,j} \in \{1, i, -1, -i\}$. Hence, $Z$ is quad-phase. \\
Case 2: $ A_{i,j} = \pm 1$ and $B_{i,j}= 0$.  Then $Z_{i,j}= A_{i,j} +iB_{i,j} \Rightarrow Z_{i,j} = \pm 1 +i0 \Rightarrow Z_{i,j} \in \{1, i, -1, -i\}$.  Again, this proves that $Z$ is quad-phase. \\
Next, we prove that $ZZ^*$ is diagonally regular by first calculating
\begin{align*}
ZZ^*&=(A+iB)(A+iB)^*\\
&=(A+iB)(A^*-iB^*)\\
&=AA^*+i(BA^*-B^*A)+BB^*
\end{align*}
Since A and B are ternary CCMs, $AA^*$ and $BB^*$ are diagonally regular.  Moreover, by assumption $BA^*-B^*A$ is also diagonally regular.  Thus, $ZZ^*$ must be diagonally regular.
\end{proof}

In the following table, a large portion, and in some cases all, of our CCM equivalence classes can be represented as a dual pair, which has a smaller search space.  This is represented in the following table that compares the results of implementing different types of construction methods.  

\begin{table}[H]
\caption{CCM Equivalence Classes}
\label{table:construction-methods}
\centering
\begin{tabular}{c|c|c|c|c}
\hline\hline
$N\times K$ p-phase CCMs & Equivalence Classes & $A\oplus B$ & $[A,B]$ & Dual Pair Representatives\\
\hline\\[-1.0em]
2x4 4-phase & 2 &n/a &2  &2\\
3x4 4-phase & 5 &  n/a& 1 &5\\
4x4 4-phase & 24 & 2 & 6 &22 \\
5x4 4-phase & 133 & n/a &  3  &94\\
6x4 4-phase & 1448 &2&27& 471\\
\hline\hline
\end{tabular}
\end{table}

\section{Acknowledgements}
    The authors would like to give special thanks to Gregory Coxson for his guidance and help over the months that this research has taken place.  We also thank Jon Russo for his wonderful search algorithm and Long (Winston) Cheong for his assistance in the programming.  
    
\section{Appendix}

\begin{itemize}
\item[I.]{$2\times 4$ Equivalence Class Representations:
{\small
\begin{align*}
1.& \hphantom{111} [[1, 1, 1, 1],[1, 1, -1, -1]]\\
2.& \hphantom{111} [[1, 1, 1, 1],[1, i, -1, -i]]
\end{align*}}
\noindent \textbf{NOTE}: Both classes can be represented by dual pairs.}

\item[II.]{$3\times 4$ Equivalence Class Representations:
{\small
\begin{align*}
1.& \hphantom{111} [[1, 1, 1, 1],[1, -1, -1, -1],[1, 1, -1, -1]]\\
2.& \hphantom{111} [[1, 1, 1, 1],[1, -1, -1, -1],[1, i, -1, -i]]\\
3.& \hphantom{111} [[1, 1, 1, 1],[1, i, -1, -i],[1, -1, 1, -1]]\\
4.& \hphantom{111} [[1, 1, 1, 1],[1, -1, -1, -i],[i, i, -i, -i]]\\
5.& \hphantom{111} [[1, 1, 1, 1],[1, i, i, -1],[-1, 1, 1, -1]]
\end{align*}}
\noindent \textbf{NOTE}: All classes can be represented by dual pairs.}

\item[III.]{$4\times 4$ Equivalence Class Representations:\\
\begin{minipage}[t]{0.5\textwidth}
{\small
\begin{align*}
1.& \hphantom{111} [[1, 1, 1, 1],[1, 1, 1, 1],[1, 1, -1, -1],[-1, -1, 1, 1]]\\
2.& \hphantom{111} [[1, 1, 1, 1],[1, 1, 1, i],[-1, -1, -i, 1],[1, -1, 1, -1]]\\
3.& \hphantom{111} [[1, 1, 1, 1],[1, 1, i, i],[1, -1, i, -i],[-1, 1, 1, -1]]\\
4.& \hphantom{111} [[1, 1, 1, 1],[1, 1, i, i],[i, -i, 1, -1],[1, -1, -1, 1]]\\
5.& \hphantom{111} [[1, 1, 1, 1],[1, 1, -1, -1],[1, -1, 1, -1],[1, -1, -1, 1]]\\
6.& \hphantom{111} [[1, 1, 1, 1],[1, i, i, -1],[-1, i, i, 1],[-1, 1, 1, -1]]\\
7.& \hphantom{111} [[1, 1, 1, 1],[1, 1, 1, 1],[1, i, -1, -i],[-1, -i, 1, i]]\\
8.& \hphantom{111} [[1, 1, 1, 1],[1, 1, i, i],[1, -1, 1, -1],[-1, 1, -i, i]]\\
9.& \hphantom{111} [[1, 1, 1, 1],[1, 1, i, i],[i, -1, 1, -i],[1, -1, -i, i]]\\
10.& \hphantom{111} [[1, 1, 1, 1],[1, 1, i, i],[-1, -1, 1, 1],[1, -1, i, -i]]\\
11.& \hphantom{111} [[1, 1, 1, 1],[1, 1, i, i],[i, -i, i, -i],[1, -1, -i, i]]\\
12.& \hphantom{111} [[1, 1, 1, 1],[1, 1, -1, -1],[1, -1, i, -i],[1, -1, -i, i]]
\end{align*}}
\end{minipage}%
\begin{minipage}[t]{0.5\textwidth}
{\small
\begin{align*}
13.& \hphantom{111} [[1, 1, 1, 1],[1, 1, i, -1],[1, -i, -1, 1],[i, -1, 1, -i]]\\
14.& \hphantom{111} [[1, 1, 1, 1],[1, 1, i, -1],[1, -i, -i, i],[i, -1, 1, -i]]\\
15.& \hphantom{111} [[1, 1, 1, 1],[1, 1, 1, i],[i, -1, -i, -i],[-i, i, 1, -1]]\\
16.& \hphantom{111} [[1, 1, 1, 1],[1, 1, 1, i],[-1, -1, -i, 1],[i, -i, 1, -1]]\\
17.& \hphantom{111} [[1, 1, 1, 1],[1, 1, i, -1],[i, -1, 1, -1],[1, -i, -1, i]]\\
18.& \hphantom{111} [[1, 1, 1, 1],[1, 1, i, i],[i, -i, i, -i],[-i, i, 1, -1]]\\
19.& \hphantom{111} [[1, 1, 1, 1],[1, 1, 1, 1],[i, i, -i, -i],[-i, -i, i, i]]\\
20.& \hphantom{111} [[1, 1, 1, 1],[1, 1, -1, -1],[i, -i, i, -i],[i, -i, -i, i]]\\
21.& \hphantom{111} [[1, 1, 1, 1],[1, i, i, -1],[-i, -1, -1, i],[-i, i, i, -i]]\\
22.& \hphantom{111} [[1, 1, 1, 1],[1, i, -1, -i],[-i, i, -i, i],[-i, -1, i, 1]]\\
23.& \hphantom{111} [[1, 1, 1, 1],[1, i, -1, -i],[-1, 1, -1, 1],[-1, i, 1, -i]]\\
24.& \hphantom{111} [[1, 1, 1, 1],[1, i, -1, -i],[1, -1, 1, -1],[1, -i, -1, i]]
\end{align*}}
\end{minipage}
\noindent \textbf{NOTE}:  All classes, except classes 23 and 24, can be represented by dual pairs.}

\item[IV.]{$5\times 4$ and $6\times 4$ Equivalence Class Representations: These representations can be downloaded from our website at the following address:
http://elvis.rowan.edu/datamining/ccm/equivalence/}

\end{itemize}



\end{document}